\documentclass[11pt]{article}
\usepackage{graphics,latexsym,amsfonts}

\usepackage{amsmath}
\usepackage{amssymb}
\usepackage{latexsym}
\usepackage{epsfig}
\usepackage[mathscr]{eucal}
\usepackage{eufrak}
\usepackage{graphics} \usepackage{color} \usepackage{multicol}
\usepackage{enumerate} \usepackage{array}

\begin{document}
\newtheorem{lemma}{Lemma}
\newtheorem{definition}[lemma]{Definition}
\newtheorem{proposition}[lemma]{Proposition}
\newtheorem{theorem}[lemma]{Theorem}
\newtheorem{corollary}[lemma]{Corollary}
\newtheorem{conjecture}[lemma]{Conjecture}
\newcommand{\bx}{\hfill \rule{2mm}{2.5mm}}

\newenvironment{proof}{\emph{Proof:}}{\bx}
\def\nl{\medskip \\ \noindent}
\def\np{

\medskip}
\setcounter{page}{1}
\normalsize

\title{\vspace{-5mm}{\bf Stable project allocation under distributional constraints\footnote{A preliminary version of this paper has appeared in the proceedings of the 10th Japanese-Hungarian Symposium on Discrete Mathematics and its Applications, 2017.}}}

\author{Kolos Csaba \'Agoston$^1$, P\'eter Bir\'o$^2$\thanks{Supported by the Hungarian Academy of Sciences under its Momentum Programme (LP2016-3/2016), by OTKA grant no.\ K108673.} \ and
Rich\'ard Sz\'ant\'o$^3$
\\
\vspace{-2mm}
\\
\small
$^1$ \emph{Department of Operations Research and Actuarial Sciences,}
\vspace{-1mm}
\\
\small
\emph{Corvinus University of Budapest, H-1093, F\H{o}v\'am t\'er 13-15., Budapest, Hungary}
\vspace{-1mm}
\\
\small
\emph{Email:} {\tt kolos.agoston@uni-corvinus.hu}.
\\
\vspace{-2mm}
\\
\small
$^2$ \emph{Institute of Economics, Research Centre for Economic and Regional Studies,}
\vspace{-1mm}
\\
\small
\emph{Hungarian Academy of Sciences, H-1112, Buda\"orsi \'ut 45, Budapest, Hungary, and}
\vspace{-1mm}
\\
\small
\emph{Department of Operations Research and Actuarial Sciences, Corvinus University of Budapest}
\vspace{-1mm}
\\
\small
\emph{Email:} {\tt peter.biro@krtk.mta.hu}.
\\
\vspace{-2mm}
\\
\small
$^3$
\emph{Department of Decision Sciences}
\vspace{-1mm}
\\
\small
\emph{Corvinus University of Budapest, H-1093, F\H{o}v\'am t\'er 13-15., Budapest, Hungary}
\vspace{-1mm}
\\
\small
\emph{Email:} {\tt richard.szanto@uni-corvinus.hu}.
\\
\vspace{-2mm}
\\
\small
\vspace{-1mm}
\\
\small
\vspace{-1mm}
\\
\small
}
\date{ }
\maketitle

\vspace{-7mm}
\begin{quotation}
\small \noindent {\bf Abstract.}
In a two-sided matching market when agents on both sides have preferences the stability of the solution is typically the most important requirement. However, we may also face some distributional constraints with regard to the minimum number of assignees or the distribution of the assignees according to their types. These two requirements can be challenging to reconcile in practice. In this paper we describe two real applications, a project allocation problem and a workshop assignment problem, both involving some distributional constraints. We used integer programming techniques to find reasonably good solutions with regard to the stability and the distributional constraints. Our approach can be useful in a variety of different applications, such as resident allocation with lower quotas, controlled school choice or college admissions with affirmative action.
\end{quotation}

\begin{quote}
\end{quote}


\section{Introduction}

A centralised matching scheme has been used since 1952 in the US to allocate junior doctors to hospitals \cite{Roth84jpe}. Later, the same technology has been used in school choice programs in large cities, such as New York \cite{APR05aer} and Boston \cite{APRS05aer}. Similar schemes have been established in Europe for university admissions and school choice as well. For instance, in Hungary both the secondary school and the higher education admission schemes are organised nationwide, see \cite{www.matchinginpractice_hun_sec} and \cite{www.matchinginpractice_hun_uni}, respectively. Furthermore, it can also be used to allocate courses to students under priorities \cite{DB17ejor}. In the above mentioned applications it is common that the preferences of the applicants and the rankings of the parties on the other side are collected by a central coordinator and a so-called stable allocation is computed based on the matching algorithm of Gale and Shapley \cite{GS62amm}. Two-sided matching markets, and the above applications in particular, have been extensively studied in the last decades, see \cite{RS90} and \cite{Manlove13} for overviews from game theoretical and computational aspects, respectively.

In this paper we describe two recent applications at the Corvinus University of Budapest, where we used a similar method with some interesting caveats. In the first application we had to allocate students to projects in such a way that the number of students allocated to each project is between a lower and an upper quota, together with an additional requirement over the distribution of the foreign students. This is a natural requirement present in many applications, such as the Japanese resident allocation scheme \cite{KK14aer}. In the second application we scheduled students to companies for solving case studies in a conference, and here again we faced some distributional constraints.

We decided to use integer programming techniques for solving both applications. We had at least three reasons for choosing this technique. The first is that with IP formulations we can easily encode those distributional requirements that the organisers requested, so this solution method is robust to accommodate special features. The second reason is that the computational problem became NP-hard as the companies submitted lists with ties. Using ties in the ranking was by our recommendation to the companies, because ties give us more flexibility when finding a stable solution under the distributional constraints. We describe this issue more in detail shortly. Finally, our third reason for choosing IP techniques was that it facilitates multi-objective optimisation, e.g. finding a most-stable solution if a stable solution does not exist under the strict distributional constraints.

The usage of integer programming techniques for solving two-sided stable matching problems is very rare in the applications, and the theoretical studies on this topic have only started very recently. The reason is that the problems are relatively large in most applications, and the Gale-Shapley type heuristics are usually able to find stable solutions, even in potentially challenging cases. A classical example is the resident allocation problem with couples, which has been present in the US application for decades, and it is still solved by the Roth-Peranson heuristic \cite{RP99aer}. The underlying matching problem is NP-hard \cite{Ronn90je}, but heuristic solutions are quite successful in practice, see also \cite{BIS11jea} on the Scottish application. However, integer programming and constraints programming techniques have been developed very recently and they turned out to be powerful enough to solve large random instances \cite{BMcBM14wp}, \cite{McBT16c} and \cite{DPB15ijcai}. Similarly encouraging results have been obtained for some special college admission problems, which are present in the Hungarian higher education system. These special features also makes the problem NP-hard in general, but at least one of these challenging features, turned out to be solvable even in a real data involving more than 150,000 applicants \cite{ABMcB16}. Finally, the last paper that we highlight with regard to this topic deals with the problem of finding stable solutions in the presence of ties \cite{KM13wp}. However, we are not aware of any papers that would study IP techniques for the problem of distributional constraints.

Distributional constraints are present in many two-sided matching markets. In the Japanese resident allocation the government wants to ensure that the doctors are evenly distributed across the country, and to achieve this they imposed lower quotas on the number of doctors allocated in each region \cite{KK14aer,KK17jet,GKKTY17aea}. Distributional objectives can also appear in school choice programs, where the decision makers want to control the socio-ethnical distribution of the students \cite{AE07wp,Bo16geb,EY15aer,EHYY14jet}. Furthermore, the same kind of requirements are implemented in college admission schemes with affirmative action \cite{Abdulkadiroglu05ijgt} such as the Brazilian college admission system \cite{AB13wp} and the admission scheme to Indian engineering schools \cite{AT16wp}.

When stable solution does not exists for the strict distributional constraints then we either need to relax stability or to adjust the distributional constraints. In this study we will consider the trade-off between these two goals, and develop some reasonable solution concepts. We use integer programming technique to solve the problems arising from the two real applications and also for the simulations.

\section{Definitions and preliminaries}

Many-to-one stable matching markets have been defined in many context in the literature. In the classical college admissions problem by Gale and Shapley \cite{GS62amm} the students are matched to colleges. In the computer science literature this problem setting is typically called Hospital / Residents problem (HR), due to the National Resident Matching Program (NRMP) and other related applications. In our paper we will refer the two sets as \emph{applicants} $A=\{a_1, \dots , a_n\}$ and \emph{companies} $\{C=c_1, \dots c_m\}$. Let $u_j$ denote the upper quota of company $c_j$.

Regarding the preferences, we assume that the applicants provide strict rankings over the companies, but the companies may have ties in their rankings. This model is sometimes referred to as Hospital / Residents problem with Ties (HRT) in the computer science literature, see e.g. \cite{Manlove13}. In our context, let $r_{ij}$ denote the rank of company $c_j$ in $a_i$'s preference list, meaning that applicant $a_i$ prefers $c_j$ to $c_k$ if and only if $r_{ij}<r_{ik}$. Let $s_{ij}$ be an integer representing the score of $a_i$ by company $c_j$, meaning that $a_i$ is preferred over $a_k$ by company $c_j$ if $s_{ij}>s_{kj}$. Note that here two applicants may have the same score at a company, so $s_{ij}=s_{kj}$ is possible. Let $\bar{s}$ denote the maximum possible score at any company and let $E$ be the set of applications. A \emph{matching} is a subset of applications, where each applicant is assigned to at most one company and the number of assignees at each company is less than or equal to the upper quota. A matching is \emph{complete} if every student is allocated. A matching is said to be \emph{stable} if for any applicant-company pair not included in the matching either the applicant is matched to a more preferred company or the company filled its upper quota with applicants of the same or higher scores.

In the classical college admission problem, that we refer to as HR, a stable solution is guaranteed to exist, and the two-versions of the Gale-Shapley algorithm \cite{GS62amm} find either a student-optimal or a college optimal solution, respectively. Furthermore, this algorithm can be implemented to run in linear time in the number of applications. Moreover, the student-proposing variant was also proved to be strategyproof for the students \cite{Roth84jpe}, which means that no student can ever get a better partner by submitting false preferences. Finally, the so-called Rural Hospitals Theorem \cite{Roth86e} states that the same students are matched in every stable solution, the number of assignees does not vary across stable matchings for any college, and for the less popular colleges where the upper quota is not filled the set of assigned students is fixed.

When extending the classical college admission problem with the possibility of having ties in the colleges' rankings, that we referred to as an HRT instance, the existence of a stable solution is still guaranteed, since we can break the ties arbitrarily, and a stable solution for the strict preferences is also stable for the original ones. However, now the set of matched students and the size of the stable matchings can vary. Take just the following simple example: we have two applicants, $a_1$ and $a_2$ first applying to college $c_1$ with the same score and applicant $a_2$ also applies to college $c_2$ as her second choice. Here, if we break the tie at $c_1$ in favour of $a_1$ then we get the matching $a_1c_1, a_2c_2$, whilst if we break the tie in favour of $a_2$ then the resulting stable matching is $a_2c_1$ (thus $a_1$ is unmatched). The problem of finding a maximum size stable matching turned out to be NP-hard \cite{Manloveetal02ja}, and has been studied extensively in the computer science literature, see e.g. \cite{Manlove13}. Note that when the objective of an application is to find a maximum size stable matchings, such as the Scottish resident allocation scheme \cite{IM08joco}, then the mechanism is not stategyproof. To see this, we just have to reconsider the above example, and assume that originally $a_1$ also found $c_2$ acceptable and would ranked it second, just like $a_2$. By removing $c_2$ from her list, $a_1$ is now guaranteed to get $c_1$ is the maximum size stable solution, however, for the original true preferences $a_2$ would have an equal chance to get her first choice $c_1$.

\subsection{Introduction of lower quotas}

In our first application the organisers of the project allocations wanted to ensure a minimum number of students for each company. Similar requirements have been imposed for the Japanese regions with regard to the number of residents allocated there. In our model, we introduce a lower quota $l_j$ for each company $c_j$ and we require that in a feasible matching the number of assignees at any company is between the lower and upper quotas. Stability is defined as before. We refer to the setting with strict preferences as Hospitals / Residents problem with Lower quotas (HRL) and the case with ties is referred to as Hospitals / Residents problem with Ties and Lower Quotas (HRTL).

Regarding HRL, the Rural Hospitals Theorem implies that the existence of a stable matching that obeys both the lower an upper quotas can be decided efficiently. This is because we just find one stable matching by considering the upper quotas only, and if the lower quotas are violated then there exists no stable solution under these distributional constraints. This problem can be still solved efficiently when the sets of companies have common lower and upper quotas in a laminar system, see \cite{FK16mor}.

However, the problem of deciding the existence of a stable matching for HRTL is NP-hard. To see this, we just have to remark that the problem of finding a complete stable matching for HRT with unit quotas is also NP-hard \cite{Manloveetal02ja}, so if we require both lower and upper quotas to be equal to one for all companies then the two problems are equivalent. Furthermore, no mechanism that finds a stable matching whenever there exists one can be strategyproof.

\subsection{Adding types and distributional constraints}

In our first application, the organisers want to distribute the foreign students across the projects almost equally. In our second application, there are target numbers for the total number of Hungarian, European and other participants and there are also specific lower quotas for Hungarian students by some companies. These applications motivate our problems with applicant types and distributional constraints.

Let $\mathcal{T}=\{T^1, \dots, T^p\}$ be the set of types, where $t(a_i)$ denotes the type of applicant $a_i$. For a company $c_j$, let $l_j^k$ and $u_j^k$ denote the lower and upper quota for the number of assignees of type $T^k$. Furthermore, we may also set lower and upper quotas for any type of applicants for a set of companies. In particular, we denote the lower and upper quotas for the total number of applicants of type $T^k$ assigned in the matching by $L^k$ and $U^k$, respectively. The set of feasibility constraints for the matching is now extended with these lower and upper quotas. Yet, the original stability condition, which does not consider the types of the applicants, remains the same.

\section{Solution concepts and integer programming formulations}

In all of our formulations we use binary variables $x_{ij}\in \{0,1\}$ for each application coming from applicant $a_i$ to company $c_j$. This can be seen as a characteristic function of the matching, where $x_{ij}=1$ corresponds to the case when $a_i$ is assigned to $c_j$.

When describing the integer formulations, first we keep the stability condition fixed while we implement the set of distributional constraints. Then we investigate the ways one can relax stability or find most-stable solutions under the distributional constraints.

\subsection{Finding stable solutions under distributional constraints}

In this subsection we gradually add constraints to the model by keeping the classical stability condition.

\subsubsection*{Classical HR instance}

First we describe the basic IP formulation for HR described in \cite{BB00mp}. The feasibility of a matching can be ensured with the following two sets of constraints.

\begin{equation}
\label{eq:applicant_feasible}
\sum_{j: (a_i,c_j)\in E}x_{ij}\leq 1 \mbox{ for each } a_i\in A
\end{equation}

\begin{equation}
\label{eq:college_feasible}
\sum_{i: (a_i,c_j)\in E}x_{ij}\leq u_j \mbox{ for each } c_j\in C
\end{equation}

Note that \eqref{eq:applicant_feasible} implies that no applicant can be assigned to more than one company, and \eqref{eq:college_feasible} implies that the upper quotas of the companies are respected.\footnote{These conditions are standard for the assignment problem as well, see a survey on this problem and its variants \cite{Pentico07ejor} and an interesting application on marriage markets \cite{CFGSW10ejor}.}

To enforce the stability of a feasible matching we can use the following constraint.

\begin{equation}
\label{eq:stable}
 \left(\sum_{k: r_{ik}\leq r_{ij}} x_{ik}\right)\cdot u_j + \sum_{h: (a_h,c_j)\in E, s_{hj}>s_{ij}}x_{hj}\geq u_j \mbox{ for each }(a_i,c_j)\in E
\end{equation}

Note that for each $(a_i,c_j)\in E$, if $a_i$ is matched to $c_j$ or to a more preferred company then the first term provides the satisfaction of the inequality. Otherwise, when the first term is zero, then the second term is greater than or equal to the right hand side if and only if the places at $c_j$ are filled with applicants with higher scores.\\

Among the stable solutions we can choose the applicant-optimal one by minimising the following objective function.

$$ \sum_{(a_i,c_j)\in E}r_{ij}\cdot x_{ij}$$

\subsubsection*{Modification for HRT}

When the companies can express ties the following modified stability constraints, together with the feasibility constraints \eqref{eq:applicant_feasible} and \eqref{eq:college_feasible}, lead to stable matchings. Note that here the only difference between this and the previous constraint is that the strict inequality $s_{hj}> s_{ij}$ became weak.

\begin{equation}
\label{eq:stable_ties}
\left(\sum_{k: r_{ik}\leq r_{ij}} x_{ik}\right)\cdot u_j + \sum_{h: (a_h,c_j)\in E, s_{hj}\geq s_{ij}}x_{hj}\geq u_j \mbox{ for each }(a_i,c_j)\in E
 \end{equation}

\subsubsection*{Extension with lower quotas}

Here, we only add the lower quotas for every company.

\begin{equation}
\label{eq:college_feasible_lower}
\sum_{i: (a_i,c_j)\in E}x_{ij}\geq l_j \mbox{ for each } c_j\in C
\end{equation}

\subsubsection*{Adding distributional constraints}

As additional constraints we require the number of assignees of a particular type to be between the lower and upper quotas for that type at a company.

\begin{equation}
\label{eq:type_upper}
\sum_{i: t(a_i)=T^k, (a_i,c_j)\in E}x_{ij}\leq u_j^k \mbox{ for each } c_j\in C \mbox{ and } T^k\in\mathcal{T}
\end{equation}

\begin{equation}
\label{eq:type_lower}
\sum_{i: t(a_i)=T^k, (a_i,c_j)\in E}x_{ij}\geq l_j^k \mbox{ for each } c_j\in C \mbox{ and } T^k\in\mathcal{T}
\end{equation}

We can also add similar constraints for set of companies, or for the overall number of different assignees at all companies. We describe the latter, as we will use it when solving our second application.

\begin{equation}
\label{eq:type_upper_overall}
\sum_{i,j: t(a_i)=T^k, (a_i,c_j)\in E}x_{ij}\leq U^k  \mbox{ for each } T^k\in\mathcal{T}
\end{equation}

\begin{equation}
\label{eq:type_lower_overall}
\sum_{i,j: t(a_i)=T^k, (a_i,c_j)\in E}x_{ij}\geq L^k \mbox{ for each } T^k\in\mathcal{T}
\end{equation}

\subsection{Relaxing stability}

Adding additional constraints to the problem can cause the lack of a stable matching, even if we added some flexibility with the ties.

One way to find a most-stable solution is to introduce nonnegative deficiency variables, $d_{ij}$ for each application and add them to the left side of the stability constraint \eqref{eq:stable_ties}. By minimising the sum of these deficiencies as a first objective we can obtain a solution which is close to be stable.

\begin{equation}
\label{eq:stable_relax1}
 \left(\sum_{k: r_{ik}\leq r_{ij}} x_{ik}\right)\cdot u_j + \sum_{h: (a_h,c_j)\in E, s_{hj}\geq s_{ij}}x_{hj} + d_{ij}\geq u_j \mbox{ for each }(a_i,c_j)\in E
\end{equation}

Note that here, if a pair $(a_i,c_j)$ is blocking for the assignment then we need to add more compensation $d_{ij}$ if the number of assignees at $c_j$ that the company prefers to $a_i$ is large. This approach can be reasonable if we want to avoid the refusal of a very good candidate at a company. We call this solution as \emph{matching with minimum deficiency}.

Alternatively, if we just want to minimise the number of blocking pairs then we can set $d_{ij}$ to be binary and minimise the sum of these variables under the following modified constraints.

\begin{equation}
\label{eq:stable_relax2}
 \left(\sum_{k: r_{ik}\leq r_{ij}} x_{ik}\right)\cdot u_j + \sum_{h: (a_h,c_j)\in E, s_{hj}\geq s_{ij}}x_{hj} + d_{ij}\cdot u_j\geq u_j \mbox{ for each }(a_i,c_j)\in E
\end{equation}

Here, every blocking pair should be compensated by the same amount, so the number of blocking pairs in minimised. Note that this concept has already been studied in the literature for various models under the name of \emph{almost stable matchings}, see e.g.\ \cite{McBT16c}.

\subsection{Adjusting upper capacities, envy-free matchings}

A different way of enforcing the lowers quota is to relax stability by artificially decreasing the capacities of the companies. This was also the solution in the resident allocation scheme in Japan \cite{KK14aer}, where the government introduced artificial upper quotas for each of the hospitals, so that in each region the sum of these artificial upper bounds summed up to the target capacity for that region. In the case of our motivating example of project allocation, one simple way of achieving the lower quotas was by reducing the upper quotas at every company.

In this solution what we essentially get is a so-called \emph{envy-free matching}, studied in \cite{WuRoth16}, \cite{Yokoi17} and \cite{AB17ecms}. The matching is stable with respect to the artificial upper quotas, which means that the only blocking pairs that may occur with regard to the original upper quotas are due to the empty slots created by the difference between the original and the artificial quotas, that we call \emph{open-slot blockings}.

However, one may not want to reduce the upper quotas of the companies in the same way, perhaps some more popular companies should be allowed to have more students than the less popular ones. Furthermore, maybe the decision on which upper quotas should be reduced should be made depending on their effect of satisfying the lower quotas (or other requirements). Thus, we may not want to set the artificial upper quotas in advance, but keep them as variables, by ensuring envy-freeness in a different way. One alternative way of enforcing envy-freeness is by the following set of constraints.

\begin{equation}
\label{eq:envy-free}
\sum_{k: r_{ik}\leq r_{ij}} x_{ik}\ge  x_{hj} \hspace{1em}
\forall (a_i,c_j),(a_h,c_j)\in E, s_{ij}>s_{hj}
\end{equation}

Constraints \eqref{eq:envy-free} will ensure envy-freeness, by making sure that if applicant $a_h$ is assigned to company $c_j$ and applicant $a_i$ has higher score than $a_h$ at $c_j$ then $a_i$ must be assigned to $c_j$ or to a more preferred company.

\subsection{Within-type priorities}

So far we have only considered different approaches of relaxing stability or enlarging the set of feasible solutions in order to satisfy the distributional constraints. In this subsection we study alternative solution concepts and methods for the case when the distributional constraints are type-dependent. This is the case also in our motivating application, where special requirements are set for the foreign students assigned to the companies.

When the number of students of a type does not achieve the minimum required at a place then there are two well-known approaches. For instance in a school choice scenario, where the ratio of an socio-ethnic group should be improved (see e.g.\ \cite{AE07wp}) then one possible affirmative action is to increase the scores of that group of students as much as needed. The other usual solution is to set some reserved seats to those students (see e.g.\ \cite{AB13wp}).

In our project allocation application our requirement is to have at least one foreign student assigned to every company. If in a stable solution this condition would be violated for a company then we can try to enforce the admission of a foreign student by increasing the scores of the foreign students at this company. We call such a solution as \emph{stable matching with type-specific scores}, where the classical stability condition is required for the adjusted scores. The second approach is to devote one place at each company to foreign students. For this one seat the foreign students will have higher priority than the locals irrespective of their scores, but for the rest of the spaces the usual score-based rankings apply. We call this concept as \emph{stable matching with reserved seats for types}. Note that neither of these two concepts can always ensure that we get at least one foreign student at each company, since they may all have high scores and they may all dislike a particular company. However, this situation changes if we also allow to decrease the scores of a group of students. We will describe this case after discussing the third approach.

Finally, as a third approach, we can also extend the concept of envy-free matchings for types. We do not require any stability with regard to students of different types, but we do require envy-freeness for students of the same type. Thus the so-called \emph{within-type envy-free matchings} will be those who satisfy the following set of constraints.\footnote{This solution concept was called \emph{within-type $\succ$-compatibility} by Echenique and Yenmez \cite{EY15aer}.}

\begin{equation}
\label{eq:envy-free_type}
\sum_{k: r_{ik}\leq r_{ij}} x_{ik}\ge  x_{hj} \hspace{1em}
\forall (a_i,c_j),(a_h,c_j)\in E, s_{ij}>s_{hj}, t(a_i)=t(a_h)=T^k, T^k\in\mathcal{T}
\end{equation}

That is, if $a_i$ and $a_h$ have the same type and $a_h$ is assigned to $c_j$ then the higher ranked $a_i$ must also be assigned to $c_j$ or to a more preferred company. Note that with this modification we extend the set of feasible solutions compared to the set of envy-free matchings. Another important observation that is motivated by our project allocation problem is that under some realistic assumptions a within-type envy-free matching always exists, that we will show in the following theorem.

\begin{theorem}\label{thm:type_specific}
Suppose that all the companies are acceptable to every student and that the sum of the lower quotas with regard to each type is less than equal to the number of students of that type, and the sum of the lower quotas across types for a company is less than or equal to the upper quota of that company, then a complete within-type envy-free matching always exists and can be found efficiently.
\end{theorem}

\begin{proof}
We construct a within-type envy-free matching separately for each type and then we merge them. When considering a particular type $T^k$, we set artificial upper quotas at the companies to be equal to the type-specific lower quotas (i.e.\ $l_j^k$ for company $c_j$) and we find a stable matching $M_k$ for this type. This stable matching must exist, since we assumed that all the companies are acceptable to every student and the number of students in every type is at least as much as the sum of the lower quotas for that type. We create matching $M$ by merging the stable matchings for the types, i.e.\ $M=M_1\cup M_2\cup\dots\cup M_p$. Note that no upper quota is violated in $M$, since we assumed that the sum of the lower quotas across types for any company $c_j$ is less than equal to the upper quota of $c_j$. By the stability of $M_k$ for every type $T^k$ it follows that matching $M$ is within-type envy-free. If there is still a company $c_j$, where the overall lower quota ($l_j$) is not yet met, then we increase an artificial upper quota for some at $c_j$ so that there is still some unmatched applicants of this type. Since the total number of applicants is greater or equal to the sum of the lower quotas, we have to achieve the lower quotas at all companies in this way. Finally, it there are still some unmatched applicants then we increase some artificial upper quotas for their types one-by-one, by making sure that we never exceed any overall upper quota. At the end of this iterative process we must reach a complete within-type envy-free matching.
\end{proof}

We note that there is a closely related solution concept introduced by Yokoi \cite{Yokoi17} which results in a within-type envy-free matching when restricted to our model, that we describe in details below. The model studied in that paper is the more general so-called classified stable matching problem where each student can have several types (e.g.\ gender, field of study, nationality) and the lower and upper quotas are set for every type. When putting their more general model in our context a student $a_i$ has justified envy towards another student $a_k$ at company $c_j$ if $a_k$ is assigned to $c_j$, $a_i$ prefers $c_j$ to her assignment, $c_j$ ranks $a_i$ higher than $a_k$, and no lower and upper quota is violated for any type when replacing $a_k$ with $a_i$ at $c_j$. It is easy to see that under the assumptions of Theorem \ref{thm:type_specific} an envy-free matching always exists as defined by Yokoi and such a solution is a within-type envy-free matching according to our definitions. Finally we remark that this model of Yokoi is originated from the classified stable matching problem introduced in \cite{Huang10soda}, and further generalised in \cite{FK16mor} and \cite{Yokoi16mor}. A common feature of these papers that the laminar nature of the set requirements makes the problem polynomial time solvable. A closely related model was studied in \cite{FBL16ejor} without the laminar assumption, where the problem was proved to by NP-hard and was solved by integer programming techniques.

Let us abbreviate a \emph{complete within-type envy-free matching} as CWTEFM.
Now, we will compare this concept of CWTEFM with stable matchings with type-specific scores and observe that they are essentially the same.

\begin{theorem}\label{thm:equivalence}
Under the assumptions of Theorem \ref{thm:type_specific} a complete matching is within-type envy-free if and only if it is stable with type-specific scores.
\end{theorem}

\begin{proof}
Suppose first that $M$ is a complete stable matching with type-specific scores, we will see that $M$ is also within-type envy-free by definition. Suppose for a contradiction that there is a student $a_i$ who has justified envy against student $a_h$ of the same type at company $c_j$, i.e.\ $a_h$ is assigned to $c_j$ whilst $a_i$ has higher score at $c_j$ than $a_h$ and $a_i$ is assigned to a less preferred company. This would mean that the pair $\{a_i,c_j\}$ is blocking for the adjusted scores, since both students get the same adjustment at $c_j$, contradicting with the stability of $M$.

Suppose now that $M$ is a CWTEFM. Let us adjust the scores of the students according to their types at each company such that the weakest students admitted have the same scores across types. Matching $M$ is stable with regard to the adjusted scores, because if a student $a_i$ is not admitted to a company $c_j$ and any better place of her preference that it must be the case that her score at $c_j$ was less then or equal to the score of the weakest assigned student of the same type at $c_j$, which means that the adjusted score of $a_i$ at $c_j$ is less than or equal to the adjusted score of every assigned student at $c_j$.
\end{proof}

Instead of using the above described processes of setting type-specific artificial upper quotas or making adjustments for the scores of different types, we can also get a CWTEFM directly by an IP formulation. We shall simply use the feasibility and distributional constraints together with \eqref{eq:envy-free_type} and with an objective function maximising the number of students assigned. This approach is not just more robust than the above described two heuristics, but it has also the advantage that we can enforce additional optimality or fairness criteria. As an additional fairness criterion we may aim to minimise the envy across types. We can achieve this by adding deficiency variables to the left hand side of constraints \eqref{eq:envy-free} for students of different types, as described in \eqref{eq:envy-free_deficiency} below, and then minimising the sum of the deficiencies. We refer to this solution as Min\#E-CWTEFM, that is \emph{complete within-type envy-free matching with minimum number of envy across types}.

\begin{equation}
\label{eq:envy-free_deficiency}
\sum_{k: r_{ik}\leq r_{ij}} x_{ik}+d_{ih}^j\ge  x_{hj} \hspace{1em}
\forall (a_i,c_j),(a_h,c_j)\in E, t(a_i)\neq t(a_h)
\end{equation}

However, we may find an envy more justified, if the score difference between the two applicants involved is higher. Thus, by taking the score differences as the intensities of the envies, we can also aim to find a refined solution where the total intensities of the envies is minimised, by using the following objective function:

$$\sum (s_{ij}-s_{hj})\cdot d_{ih}^j.$$

We call the corresponding solution \emph{complete within-type envy-free matching with minimum envy intensities across types}, abbreviated as  MinEI-CWTEFM.

If the solution is still not unique then we can further refine it, by considering two additional objectives. Regarding the welfare of the students, we may want to minimise the total rank of the students, leading to a Pareto-optimal assignment for them under the constraints. We denote these solutions as MinRank-Min\#E-CWTEFM and MinRank-MinEI-CWTEFM, depending whether we minimised the number of envies or the envy intensities in the previous round. Finally, an alternative objective can be to minimise the number of blocking pairs due to open slots. This can be achieved by adding binary deficiency variables to the first term of the left side of the stability constraints, as follows.

\begin{equation}
\label{eq:stable_ties_def}
\left(\sum_{k: r_{ik}\leq r_{ij}} x_{ik} +d_{ij}\right)\cdot u_j + \sum_{h: (a_h,c_j)\in E}x_{hj}\geq u_j \mbox{ for each }(a_i,c_j)\in E
 \end{equation}

We can then minimise the sum of these deficiency variables and find a matching within the restricted solution set that minimises the number of open-slot blockings. We denote these solutions as MinOSB-Min\#E-CWTEFM and MinOSB-MinEI-CWTEFM, depending whether we minimised the number of envies or the envy intensities.

\section{First application: CEMS project allocation}

CEMS Alliance is a global co-operation of leading business schools, multinational corporations and social partners in higher education domain. These entities run together the CEMS Master in International Management (MIM) one-year graduate program that is accessible for graduate students of the partner institutions in 29 countries in five continents. During the one-year-program students spend one semester at their home institution and one semester at another partner institution somewhere abroad, and they always learn in an international environment. CEMS MIM has been ranked as a leading master program by Financial Times in recent years.

Within the framework of the MIM program each student must carry out a business project during the Spring semester accounting for 15ECTS credits (that is half of the workload of the entire semester). The consultancy-like projects are designed as real life learning experience. Business projects are done in small groups of 3-6 students in which ideally at least one student comes from a foreign school, hence business project teams are culturally diverse. Business projects are offered and supervised by the corporate partners throughout the semester and they usually last for three months.

Students learn about the business projects during a kickoff event at the beginning of the semester from company representatives and they also receive written descriptions of the projects. After the kickoff event corporate partners evaluate all students according to their CV-s, and students also rank the business projects in the same time. The school assigns students to the individual projects based on these evaluations and rankings.

At Corvinus University of Budapest the authors of this paper have been given the task of redesigning the allocation mechanism in 2016. In previous years the mechanism was a simple immediate acceptance mechanism (also known as Boston mechanism \cite{APRS05aer}), where the students submitted their CV-s to their first choice companies, the companies evaluated the candidates and then they accepted the best candidates up to their quotas and rejected the rest. The rejected students then submitted their CV-s to further companies, but those companies which have already filled their positions did not accept more applications. This mechanism was heavily criticized in the literature on school choice due to its unfairness and also because this mechanism is highly manipulable, therefore in many cities it has been replaced by other algorithms, mainly by the deferred acceptance (or Gale-Shapley) algorithm, see e.g.\ \cite{APRS05aer}.

\subsection{Solution plan}

In 2016 there were 25 students, including 20 local and 5 foreign students, and 5 companies. The initial upper quotas were set to 6 and the lower quotas were set to 4 at all companies. The programme coordinator decided to set an upper quota of 2 for the foreign students at each company to enforce diversity. In 2017 there was a slight change in the distributional criteria, the number of students allocated to each company was set to be between 3 and 6 and at least one foreign student was required to be allocated to every company.

Our first solution plan was to ask the students to rank all the companies in a strict order and to ask the companies to evaluate all the CV-s and rank the students weakly by giving them scores between 1 and 10\footnote{Most companies gave only integer scores, but some submitted half-integer scores as well, so ties indeed occurred.}. Our intension with allowing ties was to enlarge the set of stable solutions, even though we understand that this fairness concept is a bit weaker, since we may accept a student and reject another one with the same score. Allowing ties also makes the problem NP-hard already with lower quotas, as we described in the introduction. Yet, if the ties were not allowed then the set of stable (and envy-free) solutions would be much smaller and thus it would be harder to satisfy the distributional constraints.

We remark that the conditions of Theorem \ref{thm:type_specific} are satisfied for both 2016 and 2017, since all the students have to rank (and accept) all the companies and in 2017 we were required to had at least one foreign student at each company, where the number of foreign students was more than the number of companies. Therefore a complete within-type envy-free matching always existed. Within this set of solutions we decided to minimise the number of envies across types and their intensities as the primal objectives. As secondary objectives we tried to minimise the total rank and the number of open-slot blockings.

Finally, since in both years it was possible to decrease the upper quotas at all companies by one (and set them to 5 instead of 6), we also examined these solutions. This was reasonable as allocating very different numbers of students to the companies seemed problematic, especially if some of the most popular companies was forced not to fill its quota, while less popular companies did.

\subsection{Results in 2016}

The most important results of the 2016 matching run are collected in Table \ref{ta:2016}.

\begin{table}
\begin{center}
\begin{tabular}{|l|c|c|c|c|c|c|}
  \hline
  2016 & \multicolumn{5}{c|}{profiles all/foreign} & total rank \\
   \hline
Solution 1: MinRank-Stable   & 6 & 1 & 6 & 6 & 6 & 34 \\
$u_i=6$  & 1 & 0 & 0 & 2 & 2 &  \\
\hline
Solution 2: MinRank-Stable   & 6 & 2 & 6 & 6 & 5 & 35 \\
$l_i=2, u_i=6$  & 1 & 0 & 1 & 2 & 1 &  \\
\hline
Solution 3: MinRank-Stable   & 6 & 4 & 5 & 5 & 5 & 40 \\
$u_1=6, u_i=5 (i=2..5)$  & 1 & 0 & 0 & 2 & 2 &  \\
\hline
Solution 4: MinRank-Stable   & 5 & 5 & 5 & 5 & 5 & 41 \\
$u_i=5$      & 0 & 2 & 0 & 2 & 1 &  \\
  \hline
\hline
Solution 5: MinRank-EF & 5 & 4 & 6 & 6 & 4 & 38 \\
$l_i=4, u_i=6$      & 0 & 2 & 1 & 2 & 0 &  \\
\hline
Solution 6: MinOSB-EF & 6 & 4 & 6 & 5 & 4 & 39 \\
$l_i=4, u_i=6$ & 1 & 1 & 1 & 2 & 0 &  \\
  \hline
\end{tabular}
\end{center}
\caption[]{The results of the 2016 matching run with the number of all and foreign students assigned to the companies and the total rank of the students.} \label{ta:2016}
\end{table}

In 2016 the upper bound of two for the foreign students were always satisfied without considering it, so we leave out this question from the discussion and we focus only common lower quotas. We were not able to find a stable solution for the original quotas of 4-6, since one of the companies (number 2) was very unpopular and the highest number of students that we could match there in a stable solution was 2, this is Solution 2 in Table \ref{ta:2016}. (For the record, we also checked which would be the minimum rank solution among the stable ones, that is Solution 1.) Therefore we decreased the upper quotas of all companies to 5, except the most popular company (number 1) and found a stable matching with minimum total rank (Solution 3). Note that this matching is envy-free for the original quotas. Finally we considered the possibility of decreasing all the upper quotas to 5, as described in Solution 4. From the latter two solutions the decision maker decided to choose Solution 4, since it was not substantially different from Solution 3 and for the companies it seemed to be easier to communicate the common decrease of upper quotas, compared to the case when only one company has a larger number of students.

Recently, after carefully investigating the solution concepts described in this paper, we did another check on the possible results and computed Solutions 5 and 6. Solution 5 is an envy-free solution where the total rank is minimised. It was interesting to observe that the most popular company (number 1) does not fill its upper quota, leading to many open-slot blockings at that company. Solution 6 is also envy-free, but here the open-slot blockings are minimised, but this resulted in a small decrease in the total rank.

\subsection{Results in 2017}

The results of the 2017 matching run are summarised in Table \ref{ta:2017}. In 2017 the number of students was 40 among which 13 were from abroad and the number of companies was 8. Due to the higher proportion of foreign students, the organisers decided to require the allocation of at least one foreign student to each company. The initial call suggested groups of sizes between 3 and 6, but in this year also we investigated the solutions when every upper quota was decreased to 5. In the latter case the lower quotas for the foreign students were not automatically satisfied, so we found within-type envy-free solutions and then as a first objective we either minimised the number of envies across types or we minimised the intensities of the envies. As a secondary objective we tried to minimise the total rank (there was no open-slots blocking when the upper quotas were commonly set to 5).

\begin{table}
\begin{center}
\begin{tabular}{|l|c|c|c|c|c|c|c|c|c|}
  \hline
  2017 & \multicolumn{8}{c|}{profiles all/foreign} &  total rank\\
   \hline
Solution 1: MinRank-EF   & 6 & 3 & 6 & 6 & 6 & 3 & 6 & 4 & 66 \\
$l_i=3, u_i=6$  & 1 & 1 & 1 & 4 & 1 & 1 & 3 & 1 & \\
\hline
Solution 2: MinRank-Min\#E-CWTEFM       & 5 & 5 & 5 & 5 & 5 & 5 & 5 & 5 & 85 \\
$u_i=5$, wEF, min   & 1 & 3 & 2 & 3 & 1 & 1 & 1 & 1 & \\
  \hline
Solution 3: MinRank-MinEI-CWTEFM    & 5 & 5 & 5 & 5 & 5 & 5 & 5 & 5 & 105 \\
$u_i=5$  & 1 & 4 & 1 & 2 & 1 & 1 & 1 & 2 & \\
\hline
\end{tabular}
\end{center}
\caption[]{The results of the 2017 matching run with the number of all and foreign students assigned to the companies and the total rank of the students.} \label{ta:2017}
\end{table}

\vspace{0.3cm}
Solution 1 is envy-free, and the total rank is minimised. As intuitively expected, the two least popular companies have only three students allocated each and a medium popular company has four students, whilst the popular companies receive six students. Solutions 2 and 3 are both within-type envy free for upper quotas 5. Solution 2 minimises the number of envies as the first objective and then the total rank. Solution 3 minimises the intensities of the envies and then the total rank. (Note that we also computed the minimal envy solutions without requiring within-type envy freeness, and essentially we received the same two solutions.) It is interesting to know that only one justified envy was present in both Solutions 2 and 3, and the intensity of this envy was 1 in Solution 2 and $\frac{1}{2}$ in Solution 3. However, these two solutions were rather different, and Solution 2 had much smaller total rank. Thus Solution 2 was clearly found better than Solution 3 by the decision maker. When comparing the first two solutions, the decision maker selected Solution 2, due to the more balanced sizes of groups.

\subsection{Discussion, further questions}

Here we discuss our findings and possible questions for the future.

\textbf{Importance of the distributional requirements.} We have considered our distributional constraints as hard bounds, the only relaxation we tested was the common decrease of the upper quotas. However, in many applications the distributional goals are softer, and thus may be violated. For instance, in school choice the exact proportionality with regard to ethnicity or gender may be too demanding and unnecessary to satisfy, these are rather just general aims. In such situations one may insist of the stability or the envy freeness of the solution and want to satisfy the distributional constraints as much as possible. Finally, the trade-off between fairness and distributional goals may be balanced by relaxing both requirements at the same time.

\textbf{Stability versus envy-freeness.} Leaving some slots empty to satisfy the distributional constraints is a natural way to relax stability. This is also used in the Japanese resident allocation programme, where artificial upper quotas have been set to the hospitals in order to satisfy the regional lower quotas \cite{KK14aer}. However, the open-slot blocking can also be seen as unfair from both the students' and the companies' points of views, especially when a popular company has to give up an intern. Note also that the open-slot blockings are relative to the original quotas. In our application the decision maker ended up choosing solutions in both years where the upper quotas of the companies were commonly reduced by one. These solutions admit a high number of open-slot blockings with regarding the original quotas, whilst if they are envy free for the original quotas then they are also stable (with no open-slot blockings) for the decreased quotas. Thus these chosen solutions can be seen more fair from the students' point of view, as they do not regret their rejections by a company with an open slot.

\textbf{Importance of within-type envy-freeness:} In our analyses we assumed that within-type envy-freeness is an important requirement that we obeyed in all solutions. Note that in the 2017 run we also tested the solution when this requirement was relaxed and we did not find a significant difference in the solutions. It is an interesting question how important this requirement is, and the answer can depend on the actual application. If the separation of the types is significant and there is a big difference between their performance (e.g.\ regarding the ethnicity in college admission) then within-type envy-freeness can be crucial.

\textbf{Minimising the number of justified envies or their intensities.} In our 2017 run we had a significant difference between our two recommended solutions based on minimising the number of justified envies and their intensities, respectively. In our case the former solution had much better total ranking for the students, but one can easily create an example where the opposite would happen.  If the intensities of the blocking are minimised then this means that the average difference between the scores of the students who have envy towards one to another is small. This can be more acceptable than having large score differences. In fact, if the maximum score difference is not higher than one in our application, then we could say that this solution could be seen weakly stable if the scoring by the companies were less finer, say used score range 1-5 instead of the current range of 1-10.

\textbf{Strict versus weak rankings.} Using ties in the rankings of the companies was by our recommendation in order to enlarge the set of stable (or envy-free) matchings. However, in this case stability (and envy-freeness) is weaker, the rejection of a student by a company can be explained by the admittance of another student with equal score or higher. Thus, this can be seen unfair by the student rejected, therefore in many applications (e.g.\ school choice in New York, Boston and college admissions in Ireland and Turkey) the ties are broken by lotteries or by other random factors. Ties make the problem of satisfying lower quotas NP-hard, whilst this is a polynomial-time solvable problem for strict rankings, see e.g.\ \cite{FK16mor}. Furthermore, the mechanism can become highly manipulable by the students for ties depending on the goals of the optimisation.

\textbf{Incentive issues.} A mechanism is strategy-proof for the students if neither of them can get a better match by submitting false preferences. This property holds for the student-proposing deferred-acceptance mechanism in the classical college admission model of Gale and Shapley (see e.g.\ \cite{RS90}). Strategy-proofness can also be satisfied by modified variants of the deferred-acceptance mechanism for the case of lower quotas, as suggested also for the Japanese resident allocations \cite{KK14aer,KK17jet,GKKTY17aea}. However, if we allow ties and we consider goals such as rank-minimisation then our mechanism becomes manipulable. A simple manipulation strategy for a medium-strong student can be to put her top choice as first choice, but instead of putting her true second choice in the second slot she can put some companies which are not achievable for her in any stable solution. If there is another student with the very same score and very same preferences submitting her true preferences and there is only one place left at their most preferred company then the rank-maximising algorithm with assign the manipulating student there, and the truth-telling student to the second company, since the alternative solution by exchanging the two students would result in higher total rank. Despite of this issue of manipulability, we believe that the expected gains of manipulations are negligible and their risks can be high, so in a Bayesien sense it is unlikely that a student could get a positive expected gain by manipulating. However, we admit that this hypothesis would be very hard to prove formally.

\textbf{Bound the length of preference lists.} In 2016 there were 25 students and 5 companies, in 2017 there were 40 students and 8 companies, so the screening costs of the companies have increased a lot. If this tendency will continue then the organisers of the programme may need to reconsider the requirement of providing full rankings. A reasonable solution in such situations is to have two rounds. In the first the student are required to rank a fixed number of companies, say five), and it is not guaranteed that all the students can be allocated to acceptable companies that they ranked. In the second round either no preferences are asked from the students or the organisers can elicit the preferences of the unmatched students over the companies with remaining positions. This is a standard technique also in school choice (e.g.\ in New York \cite{APR05aer}), although here we would face new challenges to ensure the satisfaction of the distributional requirements.

\section{Second application: Workshop assignment}

After running the 2016 project allocation, we received very positive feedbacks from the students, and in fact two students approached us asking for a help in selecting and assigning conference participants to companies involved in a case study workshop within the conference.

The number of participants to be selected was 60, and they had to be assigned to three companies in a given proportion, the first company had to receive 16 students and at least 8 Hungarians, the second and the third companies had to receive 22 students each. There were 13 pre-selected students (the country leaders of the organisation) whose assignments were fixed in advance, so we only had to select and assign the 47 remaining slots.

The conference organisers also agreed on the proportion of the local, regional and other students to be selected. In particular, we had to select 25 Hungarian students from the 29 Hungarian applicants, further 12 regional students from the 15 regional applicants (outside Hungary) and 10 other students from the 19 other applicants (outside the region). Thus, we had overall exact quotas (i.e.\ equal lower and upper quotas) for each type of students, just as described in \ref{eq:type_upper_overall} and \ref{eq:type_lower_overall}.

In order to satisfy these requirements we thought that we not only try to keep the solution within-type envy-free, which is also stable with type-specific scores as proved in Theorem \ref{thm:equivalence}, but we tried to keep the extra scores given to each type of students be the same across companies. We call this solution concept a \emph{stable matching with equal type-specific scores}. With an iterative testing we could indeed find such a stable solution by adding 7 extra points to all Hungarian students, 3 extra points to all regional students, and zero to the other students, where the students had to rank all the three companies and the companies gave scores (1-10) on all the applicants.

It is an interesting question whether a stable matching with equal type-specific scores always exists in our model, under the assumption that all pairs are acceptable. We state this as a conjecture below and prove it for two types.

\begin{conjecture}\label{con:equal}
When all the pairs are acceptable then a stable matching with equal type-specific scores always exists for exact quotas.
\end{conjecture}

To prove the conjecture for two types, we will use some well-known theorems listed below.

\begin{theorem}\label{thn:old}
[Well-known results on HR/HRT instances]
\begin{enumerate}[i)]
  \item (Characterisation, see e.g.\ \cite{IM08joco}) A matching $\mu$ is weakly stable for an instance $I$ of HRT if and only if it is stable for an instance $I'$ of HR that is obtained by some tie-breaking from $I$.
  \item (Rural hospitals \cite{Roth86e}) For an instance $I$ of HR the set of allocated students and the number of seats filled at the companies are fixed across the stable matchings.
  \item (Vacancy chains \cite{BRR97jet}) Suppose that $I$ is an instance of HR. If $I'$ is obtained from $I$ by adding a new student $a_i$ then the set of allocated students either 1) remains the same, 2) it is extended by $a_i$, or 3) it is extended by $a_i$ and another student, $a_j$ becomes unallocated. If $I'$ is obtained from $I$ by increasing the upper quota of a company then the set of allocated students either 1) remains the same, or 2) it is extended by one student.
\end{enumerate}
\end{theorem}

\begin{proof}[of Conjecture \ref{con:equal} for two types]
Suppose that we have a project allocation problem with two types of students, $A_1=\{a_1, a_2, \dots a_{n_1}\}$ and $A_2=\{a_{n_1+1}, a_{n_1+2}, \dots, a_{n_1+n_2}\}$, companies $C=\{c_1, \dots , c_m\}$ and exact quotas $L^1=U^1$ and $L^2=U^2$. Let $U$ denote the total capacity, i.e., $U=\sum_{j=1..m}u_j$. Without loss of generality we suppose that $L^1+L^2=U$ with $L^1\leq n_1$ and $L^2\leq n_2$, and $m\leq n_1+n_2$, which implies that all the weakly stable matchings have size $U$, since we assume that every student-company pair is mutually acceptable. Let $e$ denote the extra score given to students of the first type. Note that if $e$ is a high number then all the first type students are admitted up to the total quota and if $e$ is very small (negative) then all the second type student are admitted up to the total quota. However, the number of first type students admitted does not necessarily grow when we increase $e$.\footnote{The following example illustrates this. Let us have $A^1=\{a_1, a_2, a_3\}$ and $A^2=\{a_4,a_5\}$, and $C=\{c_1,c_2,c_3\}$ with quota 1 at each company. Suppose that every student prefers $c_1$ to $c_2$ and $c_2$ to $c_3$. The students have the following scores: $s_{1,1}=5$, $s_{1,2}=7$, $s_{1,3}=1$, $s_{2,1}=1$, $s_{2,2}=1$, $s_{2,3}=3$, $s_{3,1}=1$, $s_{3,2}=1$, $s_{3,3}=1$,  $s_{4,1}=6$, $s_{4,2}=1$, $s_{4,3}=6$,  $s_{5,1}=2$, $s_{5,2}=6$, $s_{5,3}=2$. When no extra score is added to $A^1$ then the unique stable matching is $M=\{a_1c_2, a_2c_3, a_4c_1\}$, and when we increase the score of students in $A^1$ then the unique stable matching is $M'=\{a_1c_1, a_4c_3, a_5c_2\}$, thus the number of first type students allocated decreases.} Let $I$ denote the original instance of HRT and let $I_e$ be the instance of HRT obtained after adding point $e$ to all students in $A^1$. For a matching $M$ let $|M|_1$ denote the number of first type student allocated and similarly, let $|M|_2$ denote the number of second type students allocated in $M$. The goal is to find a suitable extra score $e$ such that there exists a weakly stable matching $M$ for $I_e$ such that $|M|_1=L^1$ (which implies that $|M|_2=L^2$). In fact, we will construct a HR instance $I_e'$, obtained from $I_e$ by tie-breaking, such that matching $M$ is stable for $I_e'$ with the required distributional property.

As we already noted, if $e$ is a large negative number then $|M|_1=0$ for any stable matching $M$ in $I_e$ and if $e$ is a very large positive number then $|M|_1=min\{n_1,m\}$. For instance $I_e$ of HRT, let $I_e^{<1}$ denote the HR instance where all the ties are broken in favour of $A^1$ students (and among the students of the same type we use an arbitrary tie-breaking, say, according to their indices). Similarly, let $I_e^{<2}$ denote the HR instance, where we break all the ties in favour of $A^2$ students. First, we have to observe that $I_e^{<1}$ is the same as $I_{e+1}^{<2}$ for any $e$. Note also that the number of allocated first type students (and their set) is fixed for any HR instance by Theorem \ref{thn:old}/ii) across all stable matchings. Therefore there must exist a number $e$ such that for any stable matching $M_{e-1}$ for instance $I_{e-1}^{<1}(=I_e^{<2})$ we have $|M_{e-1}|_1\leq L^1$, and for any stable matching $M_{e}$ for instance $I_e^{<1}$ we have $|M_{e}|_1\geq L^1$. We will show that there is an instance $I_e'$, obtained by tie-breaking from $I_e$, such that for every stable matching the number of allocated $A^1$ students is exactly $L^1$.

We start from $I_{e}^{<2}$ and we will gradually transform it into $I_e^{<1}$ by giving higher priority in the tie-breaking to one $A^1$ student in each step. Let $I_{e}^{0}=I_{e}^{<2}$, and for each $i\in [1..n_1]$ let $I_{e}^{i}$ denote the HR instance where we favour the students $\{a_1,\dots , a_i\}$ over students in $A^2$, who are favoured to student in $\{a_{i+1}, \dots , a_{n_1}\}$. What we will show is that if $M$ is any stable matching for $I_{e}^{i}$ and $M'$ is any stable matching for $I_e^{i+1}$ then $|M|_1-1\leq |M'|\leq |M|_1+1$, so the number of $A^1$ students allocated can either increase or decrease by at most one. This will imply that we must get an instance $I_{e}^{i}$, where the number of first type students is exactly $L^1$.

To prove the above inequalities we have to consider two situations. First, let us assume that $a_{i+1}$ is unmatched in $M$ (and so in every stable matching for $I_e^i$). Thus $M$ would also be stable if we would remove $a_{i+1}$ from $I_e^i$. Let us now put back $a_{i+1}$, but with higher priority, creating instance $I_e^{i+1}$. By Theorem \ref{thn:old}/iii) either the number of $A^1$ students remain the same or it increases by one. Suppose now that $a_{i+1}$ is allocated in $M$ to company $c_j$. $M$ will remain stable if we remove both $a_{i+1}$ and one seat at $c_j$ from $I_e^i$, while the number of $A^1$ students allocated in the reduced matching decreases by one. If we add back one seat at $c_j$ and subsequently we also add back $a_{i+1}$ with increased priority, creating instance $I_e^{i+1}$, then from Theorem \ref{thn:old}/iii) we know that in each of these two steps the number of $A^1$ students matched either remains the same or increases by one. So, in overall, the number of $A^1$ students can either decrease by one, remain the same, or increase by one. This completes our proof.
\end{proof}

\section{Conclusion}

We investigated different solution concepts for stable matching problems with distributional constraints motivated by two real applications where we had to design the allocation mechanism. We chose integer programming as the solution technique which proved to be successful for these relatively small applications. We believe that our solution concepts and techniques could be considered in other applications as well, such as controlled school choice and university admission with affirmative action. As far as the participants are concerned, we have received very positive feedbacks from both the students and the companies, especially compared to the previous years. There are still plenty of interesting questions to investigate mostly about the importance of different fairness criteria and the trade-off between fairness and the distributional requirements.


\begin{thebibliography}{1}

\bibitem{Abdulkadiroglu05ijgt}
Abdulkadiroglu A..
\newblock College admissions with affirmative action.
\newblock {\em International Journal of Game Theory}, 33(4):535--549, 2005.

\bibitem{AE07wp}
Abdulkadiroglu A., and Ehlers L.
\newblock Controlled School Choice.
\newblock {\em Working paper}, 2007.

\bibitem{APR05aer}
Abdulkadiro{\u{g}}lu A., Pathak A., and Roth A.E.
\newblock The {N}ew {Y}ork {C}ity high school match.
\newblock {\em American Economic Review, Papers and Proceedings},
  95(2):364--367, 2005.

\bibitem{APRS05aer}
Abdulkadiro{\u{g}}lu A., Pathak P.A., Roth A.E., and S{\"o}nmez T.
\newblock The {B}oston public school match.
\newblock {\em American Economic Review, Papers and Proceedings},
  95(2):368--371, 2005.

\bibitem{AB17ecms}
\'Agoston K.Cs., and Bir\'o P.
\newblock Modelling Preference Ties and Equal Treatment Policy.
\newblock {\em In Proceedings of ECMS 2017: 31st European Conference on Modelling and Simulation}, pages 516-522, 2017.

\bibitem{ABMcB16}
\'Agoston K.Cs., Bir\'o P., and McBride I.
\newblock Integer programming methods for special college admissions problems.
\newblock {\em Journal of Combinatorial Optimization}, 32(4):1371--1399, 2016.

\bibitem{AB13wp}
Ayg\"un O., and Bo I.
\newblock College Admission with Multidimensional Reserves: The Brazilian Affirmative Action Case.
\newblock {\em Working paper}, 2013.

\bibitem{AT16wp}
Ayg\"un O., and Turhan B.
\newblock Dynamic reserves in matching markets: Theory and applications.
\newblock {\em Working paper}, 2016.

\bibitem{BB00mp}
Ba\"iou M., and Balinski M.
\newblock The stable admissions polytope.
\newblock {\em Mathematical Programming}, 87(3), Ser. A:427--439, 2000.

\bibitem{www.matchinginpractice_hun_sec}
Bir\'o P.
\newblock {M}atching {P}ractices for {S}econdary {S}chools -- {H}ungary.
\newblock  {\texttt matching-in-practice.eu}, accessed on 23 August 2014

\bibitem{www.matchinginpractice_hun_uni}
Bir\'o P.
\newblock University admission practices - Hungary.
\newblock  {\texttt matching-in-practice.eu}, accessed on 23 August 2014

\bibitem{BIS11jea}
Bir\'o P., Irving R.W., and Schlotter I.
\newblock Stable matching with couples -- an empirical study.
\newblock {\em ACM Journal of Experimental Algorithmics}, 16, Article No.: 1.2, 2011.

\bibitem{BMcBM14wp}
Bir\'o P., McBride I., and Manlove D.F.
\newblock The Hospitals / Residents problem with Couples: Complexity and Integer Programming models.
\newblock {\em In Proceedings of SEA 2014: the 13th International Symposium on Experimental Algorithms}, vol 8504 of LNCS, pp 10-21, Springer, 2014.

\bibitem{BRR97jet}
Blum J., Roth A.E., and Rothblum U.G.
\newblock Vacancy chains and equilibration in senior-level labor markets.
\newblock {\em ACM Journal of Experimental Algorithmics}, 76(2):362--411, 1997.

\bibitem{Bo16geb}
Bo I.
\newblock Fair implementation of diversity in school choice.
\newblock {\em Games and Economic Behavior}, 97:54--63, 2016.

\bibitem{McBT16c}
McBride I., Manlove D.F., and Trimble J.
\newblock ''Almost stable'' matchings in the Hospitals / Residents problem with Couples.
\newblock {\em Constraints}, 22(1):50--72, 2016.

\bibitem{CFGSW10ejor}
Cao N.V., Fragni\'ere E., Gautier J.-A., Sapin M., and Widmer E.D.
\newblock Optimizing the marriage market: An application of the linear assignment model.
\newblock {\em European Journal of Operational Research}, 202(2):547--553, 2010.

\bibitem{DB17ejor}
Diebold F., and Bichler M.
\newblock Matching with indifferences: A comparison of algorithms in the context of course allocation.
\newblock {\em European Journal of Operational Research}, 260(1):268--282, 2017.

\bibitem{DPB15ijcai}
Drummond J., Perrault A., and Bacchus F.
\newblock SAT is an effective and complete method for solving stable matching problems with couples.
\newblock {\em In Proceedings of the Twenty-Fourth International Joint Conference on Artificial Intelligence (IJCAI-15)}, 2015.

\bibitem{EY15aer}
Echenique F., and Yenmez M.B.
\newblock How to control controlled school choice.
\newblock {\em American Economic Review}, 105(8):2679--2694, 2015.

\bibitem{EHYY14jet}
Ehlers L., Hafalir I.E., Yenmez M.B., and Yildirim M.A.
\newblock School choice with controlled choice constraints: Hard bounds versus soft bounds.
\newblock {\em Journal of Economic Theory}, 153:648--683, 2014.

\bibitem{FBL16ejor}
Firat M., Briskorn M., and Laugier A.
\newblock A {B}ranch-and-{P}rice algorithm for stable workforce assignments with hierarchical skills.
\newblock {\em European Journal of Operational Research}, 251:676--685, 2016.

\bibitem{FK16mor}
Fleiner T., and Kamiyama N.
\newblock A Matroid Approach to Stable Matchings with Lower Quotas.
\newblock {\em Mathematics of Operations Research}, 41(2):734--744, 2016.

\bibitem{GS62amm}
Gale D., and Shapley L.S.
\newblock College {A}dmissions and the {S}tability of {M}arriage.
\newblock {\em American Mathematical Monthly}, 69(1):9--15, 1962.

\bibitem{GKKTY17aea}
Goto M., Kojima F., Kurata R., Tamura A., and Yokoo M.
\newblock Designing Matching Mechanisms Under General Distributional Constraints.
\newblock {\em American Economic Journal: Microeconomics}, forthcoming, 2017.

\bibitem{Huang10soda}
Huang C.-C.
\newblock Classified stable matching.
\newblock {\em In Proceedings of SODA 2010: the twenty-first annual ACM-SIAM symposium on Discrete Algorithms}, pages 1235-1253, 2010.

\bibitem{IM08joco}
Irving R.W., and Manlove D.F.
\newblock Approximation algorithms for hard variants of the stable marriage and hospitals/residents problems.
\newblock {\em Journal of Combinatorial Optimization}, 16(3):279--292, 2008.

\bibitem{KK14aer}
Kamada Y., and Kojima F.
\newblock Efficient Matching Under Distributional Constraints: Theory and Applications.
\newblock {\em American Economic Review},
  105(1):67--99, 2014.

\bibitem{KK17jet}
Kamada Y., and Kojima F.
\newblock Stability concepts in matching under distributional constraints.
\newblock {\em Journal of Economic Theory},  168:107--142, 2017.

\bibitem{Kojima12geb}
Kojima F.
\newblock School choice: Impossibilities for affirmative action.
\newblock {\em Games and Economic Behavior},  75(2):685--693, 2012.

\bibitem{KM13wp}
Kwanashie A., and Manlove D.F.
\newblock An Integer Programming approach to the Hospitals / Residents problem with Ties.
\newblock {\em In Proceedings of OR 2013: the International Conference on Operations Research}, pages 263-269, Springer, 2014.

\bibitem{Manlove13}
Manlove D.F.
\newblock Algorithms of Matching Under Preferences.
\newblock {\em World Scientific Publishing}, 2013.

\bibitem{Manloveetal02ja}
Manlove D.F., Irving R.W., Iwama K., Miyazaki S., and Morita Y.
\newblock Hard variants of stable marriage.
\newblock {\em Theoretical Computer Science}, 276(1-2):261--279, 2002.

\bibitem{Pentico07ejor}
Pentico D.V.
\newblock Assignment problems: A golden anniversary survey.
\newblock {\em European Journal of Operational Research}, 176:774--793, 2007.

\bibitem{Ronn90je}
Ronn E.
\newblock {NP}-complete stable matching problems.
\newblock {\em Journal of Algorithms}, 11:285--304, 1990.

\bibitem{Roth84jpe}
Roth A.E.
\newblock The evolution of the labor market for medical interns and residents:
  a case study in game theory.
\newblock {\em Journal of Political Economy}, 6(4):991--1016, 1984.

\bibitem{Roth86e}
Roth A.E.
\newblock On the allocation of residents to rural hospitals: a general property of two-sided matching markets.
\newblock {\em Ecnometrica}, 54(2):425--427, 1986.

\bibitem{RP99aer}
Roth A.E., and Peranson E.
\newblock The Redesign of the Matching Market for American Physicians: Some Engineering Aspects of Economic Design.
\newblock {\em American Economic Review}, 89:748--780, 1999.

\bibitem{RS90}
Roth A.E., and Sotomayor M.A.O.
\newblock Two-sided matching: a study in game-theoretic modeling and analysis.
\newblock {\em Cambridge: Econometric Society monographs}, 1990.

\bibitem{Yokoi16mor}
Yokoi Y.
\newblock A generalized polymatroid approach to stable matchings with lower quotas.
\newblock {\em Mathematics of Operations Research}, 42(1):238--255, 2016.

\bibitem{Yokoi17}
Yokoi Y.
\newblock Envy-Free Matchings with Lower Quotas.
\newblock arXiv preprint, arXiv:1704.04888, 2017

\bibitem{WuRoth16}
Wu Q., and Roth A.E.
\newblock The lattice of envy-free matchings.
\newblock mimeo, 2016

\end{thebibliography}
\end{document}